\documentclass[prd,nofootinbib,preprintnumbers,preprint]{revtex4}
\usepackage{empheq}

\usepackage{amsmath}
\usepackage{amsfonts}
\usepackage{amssymb}
\usepackage{dsfont}
\usepackage{bbold}
\usepackage[hyperfootnotes=false]{hyperref}
\usepackage[dvipsnames]{xcolor}

\newcommand{\E}{\mathcal{E}}
\newcommand{\F}{\mathcal F}
\newcommand{\Lie}{\mathcal L}

\newcommand{\SE}{{\mathcal S}_\E}
\newcommand{\LCM}{\nabla}
\newcommand{\LCS}{\mathcal D}
\newcommand{\SFS}{\chi}

\newcounter{example}[section]
\newcounter{remark}[section]
\newcounter{theorem}[section]
\newcounter{proposition}[section]
\newcounter{lemma}[section]
\newcounter{corollary}[section]
\newcounter{definition}[section]

\def\theremark{\arabic{section}.\arabic{remark}}
\def\thetheorem{\arabic{section}.\arabic{theorem}}

\def\thedefinition{\arabic{section}.\arabic{definition}}

\renewcommand*{\email}[1]{\footnote{Electronic address: \href{mailto:#1}{\nolinkurl{#1}} }}

\newenvironment{proof}{\noindent {\textit{Proof:}}
}{\medskip}

\newenvironment{theorem}{\refstepcounter{theorem}
\medskip\noindent{\bf Theorem \thetheorem}:\em}{\medskip}

\newenvironment{definition}{\refstepcounter{definition}\medskip\noindent{\bf
Definition \thedefinition}:\em}{\medskip}

\begin{document}

\title{The geometry of massive particle surfaces}
\author{Kirill Kobialko\email{kobyalkokv@yandex.ru}}
\author{Igor Bogush\email{igbogush@gmail.com}}
\author{Dmitri Gal'tsov\email{galtsov@phys.msu.ru}}
\affiliation{Faculty of Physics, Moscow State University, 119899, Moscow, Russia}

\begin{abstract}
We propose a generalization  of Claudel, Virbhadra, and Ellis photon surfaces to the case of massive charged particles, considering a timelike hypersurface  such that any worldline of a particle with mass $m$, electric charge $q$ and fixed total energy $\E$, initially touching it, will remain in this hypersurface forever. This definition does not directly appeal to the equations of motion, but instead make use of {\em partially umbilic} nature of the surface geometry. Such an approach should be especially useful in the case of non-integrable equations of motion. It may be applied in the theory of non-thin accretion discs, and also may serve a new tool for some general problems, such as uniqueness theorems, Penrose inequalities and hidden symmetries. The condition for the stability of the worldlines is derived, which reduces to differentiation along the flow of surfaces of a certain energy. We consider a number of examples of electrovacuum and dilaton solutions, find conditions for marginally stable orbits, regions of stable or unstable spherical orbits, stable and unstable photon surfaces, and solutions satisfying the no-force condition.

\end{abstract}


\maketitle

\setcounter{page}{2}

\setcounter{equation}{0}
\setcounter{subsection}{0}

\section{Introduction}

The recent publication of an image of a black hole at the center of our Galaxy \cite{EventHorizonTelescope:2022xnr} and Galaxy M87 \cite{EventHorizonTelescope:2019dse} stimulates further interest in shadows and relativistic images of black holes as a tool for searching for new physics in the Cosmos. The shadows of black holes and images of accretion disks around them are now released in the radio and optical ranges and they are associated with the properties of null geodesics in strong gravitational fields \cite{Cunha:2018acu, Dokuchaev:2019jqq,Bronzwaer:2021lzo,Perlick:2021aok,Virbhadra:2022ybp,Virbhadra:2022iiy,Dokuchaev:2020wqk,Bogush:2022hop}. The most important role in the mathematical understanding of shadows is played by the concept of characteristic photon surfaces \cite{Claudel:2000yi,Gibbons:2016isj,Yoshino1,Yoshino:2019dty,Galtsov:2019fzq,Galtsov:2019bty,Koga:2020akc,Shoom:2017ril,Kobialko:2020vqf,Kobialko:2021uwy,Cao:2019vlu} -- timelike hypersurfaces, where bounded photon trajectories are lined up. 

Schwarzschild solution is known to posses a photon sphere \cite{Virbhadra:1999nm}. Later, this concept was generalized by Claudel, Virbhadra, and Ellis to photon surfaces \cite{Claudel:2000yi}. Any null geodesic touching this surface tangentially entirely belongs to the photon surface. An important property of the photon surfaces is established by the theorem asserting that these are timelike totally umbilic hypersurfaces in spacetime \cite{Chen}, i.e., establishes proportionality of their first and second fundamental forms. This purely geometric property can serve as a constructive definition for analyzing photon surfaces instead of solving geodesic equations and plays a decisive role in the analysis of the black hole uniqueness \cite{Cederbaum,Yazadjiev:2015hda,Yazadjiev:2015mta,Yazadjiev:2015jza,Yoshino:2016kgi,Yazadjiev:2021nfr,Koga:2020gqd,Rogatko,Cederbaumo,Cederbaum:2019rbv} and area bounds \cite{Shiromizu:2017ego,Feng:2019zzn,Yang:2019zcn}. Photon surfaces find applications in the description of optical properties and gravitational shadows \cite{Grover:2017mhm,Shipley:2019kfq,Grenzebach,Tsukamoto:2020iez,Tsukamoto:2020bjm,Tsukamoto:2021fsz,Tsukamoto:2021lpm,GrenzebachSBH,Grenzebach:2015oea,Konoplya:2021slg,Lan:2018lyj}.

Here we explore a new kind of characteristic surfaces around black holes, which have the property of holding the worldlines of massive particles, including charged ones. Such surfaces make it possible to understand the geometry of ``massive shadows'' due to such particles, which certainly exist in the vicinity of black holes (electrons, neutrinos, etc.). Such shadows are not directly observable, but their existence can be detected due to modulation of the plasma heating which will translate them into the observable ranges of electromagnetic radiation from radio to optics. 

In the case of massive particles, the characteristic surface is not a conformally invariant and not a totally umbilic hypersurface, but obeys a new partially umbilic condition. This is exactly the same geometric property that characterizes fundamental photon surfaces in stationary spacetimes \cite{Kobialko:2020vqf,Kobialko:2021uwy}. A new type of partially umbilic surfaces form spacetime foliations (locally) parameterized by the values of the energy of scattered particles and has a number of geometric properties similar to those of photon and fundamental photon surfaces. Due to this, we can expect that many classical results for photon spheres can be generalized to this case as well. In particular, we can expect existence of the restrictions on the spatial sections of such surfaces -- the so-called Penrose inequalities. Moreover, one can expect to obtain uniqueness theorems and relations to Killing tensors of second rank, which describes the hidden symmetries of spacetime \cite{Kobialko:2021aqg,Kobialko:2022ozq,Gibbons:2011hg,Cariglia:2014ysa,Frolov:2017kze}. 

The paper plan is the following. In Sec. \ref{sec:setup} we briefly describe the equations of motion for charged massive particles in spacetimes with a Killing vector and conventions of the hypersurface geometry. In Sec. \ref{sec:surfaces} we present definition of massive particle surfaces, a key theorem, and a discussion of the geometric and physical properties. In Sec. \ref{sec:examples} we apply the developed formalism to many important particular examples. The appendix contains proofs of some statements formulated in the main part of the paper.

\section{Setup} 
\label{sec:setup}

Let $M$ be a Lorentzian manifold \cite{Chen} of dimension $n\geq4$ with metric tensor $g_{\alpha\beta}$, Levi-Civita connection $\LCM_\alpha$. In addition to the metric tensor describing gravity, we introduce the electromagnetic potential $A_\alpha$ and the electromagnetic field tensor $F_{\alpha\beta}=\LCM_{[\alpha} A_{\beta]}$\footnote{We use the convention of the symmetrization and antisymmetrization with unit weight: $T_{(\alpha\beta)} = T_{\alpha\beta} + T_{\beta\alpha}$, $T_{[\alpha\beta]} = T_{\alpha\beta} - T_{\beta\alpha}$.}. The worldline $\gamma^\alpha$ of test particles with charge $q$ and mass $m$ in this geometry obeys the following equations  
\begin{equation}
   \dot{\gamma}^\alpha \LCM_{\alpha}\dot{\gamma}^\beta =q F^\beta{}_{\lambda}\dot{\gamma}^\lambda , \quad \dot{\gamma}^\alpha\dot{\gamma}_\alpha=-m^2,  \label{eq_particles}
\end{equation} 
where $\dot{\gamma}^\alpha=d \gamma^\alpha/d s$ is a four-velocity of the particle, and $s$ is an affine parameter. Case $q=0$ describes neutral massive particles, while the case $q=m=0$ describes massless neutral particles such as photons. The case of hypothetical massless charges $m=0, q\neq 0$ can also be included.

Assume that the metric $g_{\alpha\beta}$, and the electromagnetic potential $A_\alpha$ share the same symmetry with respect to the Killing vector field $k^\alpha$ \cite{Kobialko:2022ozq}, timelike in the essential part of spacetime (e.g. outside the ergosphere), i.e,
\begin{align}
   \Lie_{k} g_{\alpha\beta}=\LCM_{(\alpha} k_{\beta)}=0, \quad \Lie_{k} A_\alpha=k^\lambda \LCM_\lambda A_\alpha +  \LCM_\alpha k^\lambda  A_\lambda=0, \label{eq_killing}
\end{align}
where $\Lie_{k}$ is the Lie derivative along the vector field $k^\alpha$. Spacetimes with this type of symmetry include stationary and static geometries that are not necessarily asymptotically flat. However, for the current consideration, similarly to the study of photon surfaces, static non-rotating spaces are of primary interest. This symmetry will imply conservation of the particle total energy $\E$ defined as 
\begin{equation}\label{energy}
\E\equiv -k_\alpha\left(\dot{\gamma}^\alpha+ q A^\alpha\right).
\end{equation}
Indeed, as a consequence of Eq. (\ref{eq_killing}), the total energy $\E$ will be constant along the worldlines defined by Eq. (\ref{eq_particles}) because
\begin{align}
   d\E/ds
   = \dot{\gamma}^\alpha \LCM_\alpha \E
   & = 
    - \dot{\gamma}^\alpha \dot{\gamma}^\beta \LCM_\alpha k_\beta
    - q k_\beta F^\beta{}_{\lambda}\dot{\gamma}^\lambda
    - q \dot{\gamma}^\alpha  A_\beta \LCM_\alpha k^\beta 
    - q\dot{\gamma}^\alpha k^\beta  \LCM_\alpha A_\beta
    = \\\nonumber&= 
    - qk^\beta F_\beta{}_{\lambda}\dot{\gamma}^\lambda 
    + q \dot{\gamma}^\alpha  k^\beta \LCM_\beta  A_\alpha 
    - q\dot{\gamma}^\alpha k^\beta  \LCM_\alpha A_\beta
    = 0.
\end{align}

It is also useful consider two terms in the expression (\ref{energy}) separately, introducing the kinetic and potential energy $\E =\E_k + \E_p$:
\begin{align}
\E_p\equiv -qk_\alpha A^\alpha,\qquad
\E_k\equiv\E-\E_p.
\end{align}

In the general case, $\E_k$ and $\E_p$ are not conserved separately. The potential energy is a pre-defined function for given $k^\alpha$ and $A^\alpha$. But the kinetic energy $\E_k$ is introduced as a secondary quantity which is a certain function of coordinates for a fixed total energy. Alternatively, it can be represented as a scalar product of the  Killing vector $k^\alpha$ with some properly normalized timelike (for $m\neq 0$, null for $m=0$) vector $v^\alpha$ so that 
\begin{align} \label{eq:constraint} 
    -k_\alpha v^\alpha=\E_k=\E-\E_p, \quad v_\alpha v^\alpha=-m^2.
\end{align}
Then, the set of linearly independent vectors $v^\alpha$ will span all worldlines of the particle with fixed total energy $\E$, mass $m$ and charge $q$ passing through a given point of the spacetime.

Our main goal is to find hypersurfaces where particles with fixed total energy $\E$, mass $m$ and charge $q$ live. In order to describe such hypersurfaces, we will use the following formalism and notations similar to Ref. \cite{Kobialko:2022ozq}. Let $\mathcal{S}$ be a timelike hypersurface of dimension $n-1$ with a spacelike normal unit vector $n^\alpha$ ($n^\alpha n_\alpha=1$).  The induced hypersurface metric reads
\begin{equation}
   h_{\alpha\beta} =
    g_{\alpha\beta}
    -n_\alpha n_\beta,
\end{equation}
defining the projector operator $h^{\alpha}_{\beta}$ and the symmetric second fundamental form $\SFS_{\alpha\beta}$
\begin{equation} \label{eq:form.projector}
    h^{\alpha}_{\beta} =
    \delta^{\alpha}_{\beta}
    - n^\alpha n_\beta,\qquad
    \SFS_{\alpha\beta}
    \equiv
    h^{\lambda}_{\alpha}
    h^{\rho}_{\beta}
    \LCM_{\lambda} n_{\rho}.
\end{equation}
The corresponding tensor projections onto the tangent space of the hypersurface read
\begin{equation} \label{eq:form.projector_2}
\mathcal{T}^{\beta\ldots}_{\gamma\ldots}=h^{\beta}_{\rho}\ldots  h_{\gamma}^{\tau}\ldots T^{\rho\ldots}_{\tau\ldots}, \qquad
    \LCS_\alpha \mathcal{T}^{\beta\ldots}_{\gamma\ldots} =
    h^{\lambda}_{\alpha} h^{\beta}_{\rho}\ldots  h_{\gamma}^{\tau}\ldots \LCM_\lambda \mathcal{T}^{\rho\ldots}_{\tau\ldots}, 
\end{equation}
where $\LCS_\alpha$ is a Levi-Civita connection in $\mathcal{S}$. 

For what follows, it is useful to project the Killing vector onto the hypersurface
\begin{equation}
    k^\alpha = \kappa^\alpha+    k_\perp n^\alpha,  \qquad \kappa^\alpha n_\alpha= 0.
\end{equation}
and to distinguish cases $\kappa^2\neq0$ and $\kappa^2=0$ ($\kappa^2 \equiv \kappa_\alpha \kappa^\alpha$). For the latter case the surface $\mathcal{S}$ will represent the Killing horizon if the Killing vector is tangent to it (i.e., $k_\perp = 0$). Considering of this case will be postponed to App. \ref{sec:kh}. In the first case we can decompose vector $v^\alpha$, subject to the constraints (\ref{eq:constraint}):
\begin{align}
&v^\alpha=(-\E_k/\kappa^2) \kappa^\alpha + u^\alpha, \quad  \kappa_\alpha  u^\alpha= n_\alpha u^\alpha = 0, \quad u^2=-m^2-\E_k^2/\kappa^2,
\end{align}
where $u^\alpha$ is some vector tangent to $\mathcal{S}$ and orthogonal to $\kappa^\alpha$. In most of the spacetime $\kappa^2<0$ and absolute value $|\kappa^2|$ must satisfy an additional inequality
\begin{equation}\label{eq:restriction}
    0<|\kappa^2|\leq\E_k^2/m^2.
\end{equation}
Indeed, in a Lorentzian manifold the orthogonal complement of a timelike vector $\kappa^\alpha$ can contain only spacelike vectors, i.e., $u^2>0$ (or $u^\alpha=0$). In particular, we find $\E_k\neq0$. This general limitation on the kinetic energy is also preserved in the case $\kappa^2=0$ (see App. \ref{sec:kh}). Restriction (\ref{eq:restriction}) has a simple physical meaning in terms of the particle motion. The strict equality corresponds to the classical turning points of worldlines in $\mathcal{S}$ while the inequality specifies the areas in hypersurface allowed for motion \cite{Kobialko:2020vqf,Kobialko:2021uwy}. In the case of massless particles, as expected, there are no constraints on nonzero $\E_k$, since the right hand side of Eq. (\ref{eq:restriction}) tends to infinity for any finite $\E_k$. This corresponds to conformal invariance of null geodesic equations.

\section{Massive particle surfaces}
\label{sec:surfaces}

As it is mentioned in Ref. \cite{Claudel:2000yi}, the photon sphere $\mathcal{S}_{\rm ph}$ in Schwarzschild spacetime has two main properties: (\textit{i}) any null geodesic initially tangent to $\mathcal{S}_{\rm ph}$ will remain tangent to it; (\textit{ii}) $\mathcal{S}_{\rm ph}$ does not evolve with time. The general definition of a photon surface is based on only the first of these properties and leads to the fact that such surfaces are totally umbilic \cite{Chen}. Now we would like to give the generalization of the photon surfaces for massive charged particles of fixed total energy: the {\em massive particle surfaces}.   

\begin{definition} 
A massive particle surface in $M$ is an immersed, timelike, nowhere orthogonal to Killing vector $k^\alpha$ hypersurface $\SE$ of $M$ such that, for every point $p\in \SE$ and every vector $v^\alpha|_p \in T_p\SE$ such that $v^\alpha  \kappa_\alpha|_p=-\E_k|_p$ and $v^\alpha v_\alpha|_p=-m^2$, there exists a worldline $\gamma$ of $M$ for a particle with mass $m$, electric charge $q$ and total energy $\E$ such that $\dot{\gamma}^\alpha(0) =v^\alpha|_p$ and $\gamma\subset \SE$\footnote{The worldline is considered to lie on the surface $\SE$ in some open neighborhood of any interior point, but in general it can leave the surface through the boundary if the latter exists.}.
\end{definition} 

In other words, any worldline of a particle with mass $m$, electric charge $q$ and total energy $\E$ initially tangent to massive particle surface $\SE$ will remain tangent to $\SE$. This definition reduces to the definition of photon surfaces \cite{Claudel:2000yi} if we restrict it to the massless $m=0$, chargeless $q=0$ case and require arbitrariness of the total energy. However, as we will see later the arbitrariness of the worldline energy (or, more precisely, photon frequency) will arise automatically in this particular case. The independence of null geodesics on the photon frequency is a consequence of conformal invariance of the Maxwell theory. The theories of massive particles do not possess conformal invariance, so it is crucial that we define the characteristic surfaces for massive particles only for fixed energy. 

The non-orthogonality condition for Killing vector makes it possible to have the non-zero kinetic energy $\E_k$ for a worldline tangent to the surface $\SE$ and is in fact a natural condition for timelike surfaces. In the case when the Killing vector is tangent to the $\SE$, massive particle surfaces automatically inherits the corresponding symmetry of the original spacetime. In particular, for timelike Killing vectors, such surface do not evolve with time, i.e., also satisfy the second condition for photon spheres in the Schwarzschild spacetime. The key geometric properties of the massive particle surfaces are given by the following theorem. 

\begin{theorem}\label{th:theorem}
Let $\SE$ is an immersed, timelike hypersurface in $M$ and $k^\alpha$ is a Killing vector nowhere orthogonal to $\SE$. If $\kappa^2>-\E^2_k/m^2$ and $\E_k\neq0$ everywhere on $\SE$, the following statements are equivalent:

(i) $\SE$ is a massive particle surface for given $q, m$ and $\E$;

(ii) the second fundamental form satisfies 
\begin{equation}
    \SFS_{\alpha\beta} v^\alpha v^\beta=-q n^\beta F_{\beta\lambda} v^\lambda,
\end{equation}
for all $p\in\SE$ and $\forall v^\alpha \in T\SE$ so that $v^\alpha v_\alpha=-m^2$ and $v^\alpha \kappa_\alpha=-\E_k$;

(iii) the second fundamental form satisfies
\begin{equation}\label{eq:condition_3}
\SFS_{\alpha\beta}=\cfrac{\SFS_\tau}{n-2} H_{\alpha\beta}+(q/\E_k)\F_{\alpha\beta},
\end{equation}
where  $H_{\alpha\beta}$ is related to the induced metric $h_{\alpha\beta}$ as
\begin{align}\label{eq:definitions_3} 
    &
    H_{\alpha\beta}=h_{\alpha\beta}+ (m^2/\E_k^2)  \kappa_\alpha \kappa_\beta,\qquad
    \F_{\alpha\beta}=\frac{1}{2}n^\rho F_{\rho (\alpha} \kappa_{\beta)}, 
    \\\nonumber &  H=H^{\alpha}{}_{\alpha}, \qquad
    \SFS=\SFS^{\alpha}{}_{\alpha}, \qquad
    \F = {\F^\alpha}_\alpha = n^\rho F_{\rho\lambda} \kappa^\lambda,
    \\\nonumber&
    \SFS_\tau = \frac{n-2}{H}(\SFS - q\F/\E_k);
\end{align}

(iv) every worldline in $\SE$ with $\dot{\gamma}^\alpha \kappa_\alpha|_p=-\E_k|_p$ and $\dot{\gamma}^\alpha  \dot{\gamma}_\alpha|_p=-m^2$ at some point $p\in \SE$ is a worldline in $M$. 

\end{theorem} 

\begin{proof}
We will prove consequently that $(i) \Rightarrow (ii)$, $(ii) \Rightarrow (iii)$, $(iii) \Rightarrow (iv)$ and $(iv) \Rightarrow (i)$, so any statement implies the other one. In the proof of  $(ii)\Rightarrow(iii)$ we will use the decomposition suitable for any surface except the case $\kappa^2 = 0$. For the sake of clarity, the proof for this case is given in App. \ref{sec:kh}. 

$(i)\Rightarrow(ii)$. Suppose $\SE$ is a massive particle surface. Let $p\in \SE$ and let $v^\alpha|_p \in T_p\SE$ such that $v^\alpha  \kappa_\alpha|_p=-\E_k|_p$. Then there exists a worldline $\gamma$ ($\dot{\gamma}^\alpha\dot{\gamma}_\alpha=-m^2$) of $M$ such that $\dot{\gamma}^\alpha(0) =v^\alpha|_p$, $\gamma\subset S$. Consider the Gauss decomposition \cite{Chen} of the covariant derivative in the particle equation of motion 
\begin{align}
    q{F^\beta}_\lambda \dot{\gamma}^\lambda = \dot{\gamma}^\alpha \LCM_\alpha \dot{\gamma}^\beta =
      \dot{\gamma}^\alpha \LCS_\alpha \dot{\gamma}^\beta
    - \SFS_{\alpha\sigma} \dot{\gamma}^\alpha \dot{\gamma}^\sigma n^\beta.
\end{align}
Projecting this equation onto the normal and tangent subspaces, one obtains the following two conditions:
\begin{subequations}
\begin{align}
    & \label{eq:tangent_projection}
    q h^{\beta\gamma} F_{\gamma\lambda} \dot{\gamma}^\lambda = \dot{\gamma}^\alpha \LCS_\alpha \dot{\gamma}^\beta,
    \\ & \label{eq:normal_projection}
    q n^\beta F_{\beta\lambda} \dot{\gamma}^\lambda = - \SFS_{\alpha\beta} \dot{\gamma}^\alpha \dot{\gamma}^\beta.
\end{align}
\end{subequations}
The Eq. (\ref{eq:tangent_projection}) is an equation of motion of the charged particle in the hypersurface $\SE$, while the constraint (\ref{eq:normal_projection}) can be treated as a condition for the hypersurface itself. So, the following condition must hold for any $p \in \SE$
\begin{align} \label{eq:chi_condition}
    \SFS_{\alpha\beta} v^\alpha v^\beta=-q n^\beta F_{\beta\lambda} v^\lambda.
\end{align}

$(ii)\Rightarrow(iii)$. We can always decompose the tensors $v^\alpha$ and $\SFS_{\alpha\beta}$ as follows (for $\kappa^2 \neq 0$)
\begin{subequations}
\begin{align}\label{eq:v_decomposition}
&v^\alpha=(-\E_k/\kappa^2) \kappa^\alpha + u^\alpha, \quad  \kappa_\alpha  u^\alpha=0, \quad u^2=-m^2-\E_k^2/\kappa^2,
\end{align}
\begin{align}
&\SFS_{\alpha\beta} = \alpha \kappa_\alpha  \kappa_\beta + \kappa_{(\alpha}  \beta_{\beta)}+\lambda_{\alpha\beta}+(q/\E_k)\F_{\alpha\beta}, \quad \kappa^\alpha\lambda_{\alpha}{}_{\beta} =0, \quad  \kappa^\alpha\beta_\alpha =0,
\end{align}
\end{subequations}
where the symmetrical form $\F_{\alpha\beta}\equiv \frac{1}{2} n^\rho  F_{\rho (\alpha} \kappa_{\beta)}$ in $\chi_{\alpha\beta}$ was introduced to compensate the right hand side in Eq. (\ref{eq:chi_condition}), giving the following general condition
\begin{align} \label{eq:intermediate}
\alpha \E_k^2 - 2 \E_k  \beta_\alpha u^\alpha +\lambda_{\alpha}{}_{\beta}u^\alpha u^\beta=0.
\end{align}
Since the condition (\ref{eq:intermediate}) holds for any $u^\alpha$, and mapping $u^\alpha \to - u^\alpha$ preserves the norm of the vector $u^\alpha$, this condition must hold for any sign of $u^\alpha$
\begin{align}
\alpha \E_k^2\pm2 \E_k  \beta_\alpha u^\alpha +\lambda_{\alpha}{}_{\beta}u^\alpha u^\beta=0.
\end{align}
Adding and subtracting these equations, we find
\begin{align} \label{eq:lambda_beta}
\lambda_{\alpha}{}_{\beta}u^\alpha u^\beta=- \alpha \E_k^2,  \quad  \beta_\alpha u^\alpha=0.
\end{align}
From the arbitrariness of $u^\alpha$ we get $\beta_\alpha=0$.

Introducing an orthonormal basis $\{e_a{}^\alpha\}$ ($e_a{}^\alpha e_b{}_\alpha=\eta_{ab}$, with indices $a,b=0,1,\ldots,n-3$), the equation $\lambda_{ab}u^a u^b=- \alpha \E_k^2$ must hold for all vectors satisfying $\eta_{ab} u^a u^b = -m^2 - \E_k^2/\kappa^2$, where the matrix $\eta_{ab} = \text{diag}(\pm 1, 1, \ldots, 1)$ is a flat metric with the first element reflecting the signature of the tangent subspace (with the symmetry $G=SO(n-2)$ or $SO(n-3,1)$). One can expect, that $\lambda_{ab}$ has the same symmetry, since it should not depend on the choice of the basis. This is possible if $\lambda_{ab}$ is a unity element of $G$ up to an arbitrary multiplier $\lambda_{ab} = \lambda \eta_{ab}$ for some $\lambda$. Let us prove this more strictly.
From the full set of possible vectors $u^a$, choose a subset parameterized by $a, b$ and $\psi$
\begin{align}
&u^\alpha=a e_0{}^\alpha+ b (e_i{}^\alpha \cos \psi  + e_j{}^\alpha \sin\psi),
\end{align}
with the constraint $b^2=-\eta_{00} a^2-m^2-\E_k^2/\kappa^2$ and indices $i,j=1,\ldots,n-3$. Then the left hand side of Eq. (\ref{eq:lambda_beta}) is
\begin{align}
    \lambda_{ab}u^a u^b
    & =
      a^2\SFS_{00} 
    + \frac{b^2}{2}\left(\SFS_{ii}+ \SFS_{jj}\right)
    \\\nonumber&
    + 2 a b (\SFS_{0i}\cos \psi+\SFS_{0j}\sin\psi)
    + b^2 \left(\SFS_{ij} \sin (2\psi) + \frac{(\SFS_{ii}- \SFS_{jj})}{2}\cos (2 \psi)\right). \nonumber
\end{align}
Since angle $\psi$ and $a\neq0$ are arbitrary as long as $a$ does not violate the condition $b^2>0$ (such $a$ always exist due to the condition $\kappa^2 > -\E_k^2/m^2$) we get
\begin{align}
\SFS_{00}=\eta_{00}\SFS_{ii}, \quad  (m^2+\E_k^2/\kappa^2)\SFS_{ii}=\alpha \E_k^2, \quad \SFS_{ii}= \SFS_{jj}, \quad \SFS_{ij}=0, \quad \SFS_{0i}=0.
\end{align}
Note that for $n=4$ there are only two vectors in the tangent subspace, so one can choose $\psi=0$ or $\pi$ to stay with only two basis vectors. Taking back the holonomic basis, the final form of the tensor $\lambda_{\alpha\beta}$ is
\begin{align}
\lambda_{\alpha\beta}=\lambda (h_{\alpha\beta}-\kappa_\alpha  \kappa_\beta/\kappa^2),\qquad
\lambda \equiv \alpha \E_k^2/(m^2+\E_k^2 /\kappa^2).
\end{align}
Collecting all together, the second fundamental form can be presented as
\begin{align}
&\SFS_{\alpha\beta} 
= \alpha \kappa_\alpha  \kappa_\beta + \lambda (h_{\alpha\beta}-\kappa_\alpha  \kappa_\beta/\kappa^2)
+(q/\E_k)\F_{\alpha\beta}
=\frac{\SFS_\tau}{n-2} H_{\alpha\beta}
+(q/\E_k)\F_{\alpha\beta},
\end{align}
where the function $\lambda$ is changed to $\SFS_\tau \equiv (n-2)\lambda$.

$(iii)\Rightarrow(iv)$. If the condition $(iii)$ holds, then Eq. (\ref{eq:normal_projection}) holds as well
\begin{align}
   &
   \SFS_{\alpha\sigma} \dot{\gamma}^\alpha \dot{\gamma}^\sigma
   = \frac{\SFS_\tau}{n-2} (-m^2+(m/\E_k)^2 \E_k^2)
   -q n^\beta F_{\beta\lambda} \dot{\gamma}^\lambda
   =-q n^\beta F_{\beta\lambda} \dot{\gamma}^\lambda.
\end{align}
On the other hand, Eq. (\ref{eq:tangent_projection}) is an equation of motion for the charged particle in $\SE$, so $\gamma$ is a trajectory of the charged particle both in $M$ and $\SE$.

$(iv)\Rightarrow(i)$ Let $p \in \SE$ and let $v^\alpha|_p \in T_p\SE$ such that $v^\alpha v_\alpha|_p =-m^2$ and $v^\alpha  \kappa_\alpha|_p =-\E_k|_p $. Let $\gamma$ be a worldline of $\SE$ such that $\dot{\gamma}^\alpha(0) = v^\alpha|_p$. Then, by ($iv$), $\gamma$ is a worldline
of $M$ such that $v^\alpha v_\alpha|_p =-m^2$ and $v^\alpha  \kappa_\alpha|_p =-\E_k|_p $ and $\gamma \subset \SE$.
\end{proof}

$\Box$

The obtained result is a complete analogue of theorem 2.2 obtained in Ref. \cite{Claudel:2000yi} for photon surfaces. In particular, statement ($iii$) of the theorem describes the pure geometry of a massive particle surface without references to the worldline equations and it represents an analog of the totally umbilic condition for photon surfaces \cite{Chen,Claudel:2000yi}. This equivalent definition is an effective way to analyze surface geometry in non-integrable dynamical systems \cite{Frolov:2017kze,Kobialko:2020vqf,Kobialko:2021uwy} and to study general theoretical problems such as Penrose inequalities \cite{Shiromizu:2017ego,Feng:2019zzn,Yang:2019zcn}, uniqueness theorems \cite{Cederbaum,Yazadjiev:2015hda,Yazadjiev:2015mta,Yazadjiev:2015jza,Yoshino:2016kgi,Yazadjiev:2021nfr,Koga:2020gqd,Rogatko,Cederbaumo,Cederbaum:2019rbv} and hidden symmetries \cite{Gibbons:2011hg,Cariglia:2014ysa,Frolov:2017kze,Kobialko:2021aqg,Kobialko:2022ozq}.

The theorem works for the timelike and spacelike Killing vectors and their projections $\kappa^\alpha$. In the latter case, the quantity $\E$ may not always be associated with the energy of the particle, but is related to some other conserved quantity. This freedom allows us to analyze both surfaces that fall inside ergoregions where energy-related Killing vector is spacelike and surfaces in general geometries such as Kaluza-Klein's models \cite{Rasheed:1995zv}. As an example, if one chooses $k^\alpha\partial_\alpha=\partial_\phi$ (here, $\phi$ is an azimuth angle), then the massive particle surface would correspond to particles with fixed mass $m$, charge $q$ and angular momentum projection $L_z$. 

Let us analyze some properties of the massive particle surfaces and caught worldlines following from theorem. 

{\bf Killing basis.} For the clearer interpretation of the geometric definition ($iii$), we introduce a basis associated with spacetime symmetries. Namely, a set of the basis vectors is composed by the vector $\kappa^\alpha$ and a set of linearly independent $n-2$ vectors $\tau^\alpha_{(i)}$ in $\SE$ orthogonal to $\kappa^\alpha$ (i.e., $\tau^\alpha_{(i)} \kappa_\alpha = 0$). The projections of the second fundamental form $\chi_{\alpha\beta}$ onto vectors $\kappa^\alpha$ and $\tau_{(i)}^\alpha$ read 
\begin{subequations}\label{eq:chi_all}
\begin{align}
    & \label{eq:chi_kk}
    \kappa^\alpha\kappa^\beta\SFS_{\alpha\beta} =
    \kappa^2 \left(\SFS - \SFS_\tau\right),
    \\& \label{eq:chi_kt}
    \kappa^\alpha\tau_{(i)}^\beta\SFS_{\alpha\beta}=
    \frac{1}{2}(q/\E_k)
         \kappa^2 n^\rho F_{\rho \lambda}\tau^\lambda_{(i)},
    \\& \label{eq:chi_tt}
    \tau_{(i)}^\alpha \tau_{(j)}^\beta \SFS_{\alpha\beta} = 
    \frac{\SFS_\tau}{n-2}\tau^\alpha_{(i)}\tau^\beta_{(j)}h_{\alpha\beta},
\end{align}
\end{subequations}
and the traceless part of the second fundamental form is
\begin{align}
 &\sigma_{\alpha\beta} \equiv\SFS_{\alpha\beta}-\SFS/(n-1) h_{\alpha\beta} = 
 \frac{m}{\E_k}\kappa_{\alpha\beta}^\lambda
 \left(
 \frac{\SFS_\tau m}{(n-2)\E_k} \kappa_\lambda + \frac{q}{m} n^\rho F_{\rho\lambda}
 \right),
 \\\nonumber&
 \kappa_{\alpha\beta}^\lambda\equiv\kappa_{\alpha} h^\lambda_{\beta} + \kappa_{\beta} h^\lambda_{\alpha} - \kappa^\lambda h_{\alpha\beta} / (n-1).
\end{align}
Using the definition of Ref. \cite{Kobialko:2020vqf,Kobialko:2021uwy} for the {\em partially umbilic surfaces}, we can claim that the second fundamental form of the massive particle surface is partially umbilic with respect to directions $\tau^\alpha_{(i)}$ orthogonal to tangential projection of the Killing vector $\kappa^\alpha$. The quantity $\SFS_\tau/(n-2)$ has a meaning of the mean curvature of the orthogonal complement to $\kappa^\alpha$. Partially umbilicity of the surfaces is exactly the same property that defines fundamental photon surfaces \cite{Kobialko:2020vqf,Kobialko:2021aqg}. However, now the mixed components are not arbitrary but influenced by the Maxwell form. This geometric similarity of massive and massless surfaces is quite expected, since the construction of fundamental photon surfaces is based on the condition of capturing null geodesics with a fixed impact parameter, the role of which in the current considerations is played by energy. In particular, one can expect that massive particle surfaces form some foliations of spacetime (or its parts) parameterized by energy values, similarly to fundamental photon surfaces forming photon regions \cite{Grenzebach,Grenzebach:2015oea} with a continuous change of the impact parameter \cite{Kobialko:2020vqf}.

{\bf Principal curvatures.} In the geometric definition of massive particle surface ($iii$) we introduced a new symmetric quadratic forms $H_{\alpha\beta}$ and $\F_{\alpha\beta}$ associated with gravitational and electromagnetic fields respectively. The quadratic form $H_{\alpha\beta}$ is the induced metric with an addition of one distinguished component along the projections of the Killing vector field. For neutral particles, this additional component leads to the fact that the direction of the Killing vector $\kappa^\alpha$ defines a unique principal direction along which the principal curvature differs from all others. Indeed, from (\ref{eq:chi_kk}) and (\ref{eq:chi_tt}) the principal curvatures along directions $\tau^\alpha_{(i)}$ and $\kappa^\alpha$ for neutral particles reads
\begin{equation}
\lambda_{\tau_{(i)}}=\cfrac{\SFS_\tau}{n-2}, \qquad \lambda_\kappa= \lambda_{\tau_{(i)}} \left(1 + (m^2/\E^2)  \kappa^2\right).
\end{equation}
Thus, the surface is generally not totally umbilic (i.e., it has not equal principal curvatures) unlike the photon surface. Furthermore, under conformal transformations of the metric, the ratio of principal curvatures $\lambda_\kappa / \lambda_{\tau_{(i)}}$ changes. The requirement of the conformal invariance of this ratio leads to the condition $m=0$, i.e., to the coincidence of all principal curvatures and their independence from the energy scale, which was expected for the photon surfaces. The quadratic form $\F_{\alpha\beta}$ determines the electromagnetic force acting on the charged particles lying on the surface in the direction normal to the surface. It consists of the normal electric field $\F$ and tangential magnetic field $n^\rho F_{\rho \lambda}\tau^\lambda_{(i)}$ (defined with respect to the observer moving along $\kappa^\alpha$). In the case of a non-zero tangential magnetic field $n^\rho F_{\rho \lambda}\tau^\lambda_{(i)}$, the surface may not have a well defined principal curvatures as far as the induced metric and the second fundamental form may not be simultaneously diagonalizable \cite{Chen}. In order to diagonalize the metric and the second fundamental form simultaneously, one should find an orthogonal basis such that any two different basis vectors contracted with $\chi_{\alpha\beta}$ gives zero. For the sake of clarity we will assume that $\kappa^2<0$. First of all, we extract $n-3$ orthogonal unit vectors from the set $\tau^\alpha_{(i)}$, which are also orthogonal to $\mathfrak{F}^\beta$, where $\mathfrak{F}^\beta \equiv n^\rho F_{\rho\lambda}(h^{\lambda\beta} - \kappa^\lambda \kappa^\beta/\kappa^2)$. The remaining subspace is a linear span of the vectors $\kappa^\alpha/\sqrt{|\kappa^2|}$ and $\mathfrak{F}^\beta / \sqrt{|\mathfrak{F}^2|}$. Using Eq. (\ref{eq:chi_all}), the second fundamental form in this basis has the form
\begin{equation}
        \chi_{ab} = \begin{pmatrix}
        \tilde{\chi}_{ab} & \mathbb{0}_{n-3}
        \\
        \mathbb{0}_{n-3} & \frac{\chi_\tau}{n-2}\mathbb{1}_{n-3}
    \end{pmatrix},\qquad
    \tilde{\chi}_{ab} = \begin{pmatrix}
        \chi_\tau - \chi & -\frac{q \sqrt{|\kappa^2 \mathfrak{F}^2|}}{2\E_k}
        \\
        -\frac{q \sqrt{|\kappa^2 \mathfrak{F}^2|}}{2\E_k} & \frac{\chi_\tau}{n-2}
    \end{pmatrix},
\end{equation}
where the matrix $\tilde{\chi}_{ab}$ is a projection onto the subspace $\left(\kappa^\alpha/\sqrt{|\kappa^2|},\,\mathfrak{F}^\beta / \sqrt{|\mathfrak{F}^2|}\right)$.
Eigenvectors of the matrix ${\tilde{\chi}_a}^b$ are always orthogonal and they give the basis, where the matrix $\tilde{\chi}_{ab}$ is diagonal. They exist as long as the eigenvalues $\lambda_\pm$ of the matrix ${\tilde{\chi}_a}^b$ are not complex
\begin{equation}
    \lambda_\pm = \frac{1}{2}\left(
        \frac{2 + (m/\E_k)^2\kappa^2}{n-2}\SFS_\tau - q\F/\E_k
        \pm \sqrt{\lambda}
    \right),\qquad
    \lambda = \left(\frac{(m/\E_k)^2\kappa^2}{n-2}\SFS_\tau - q\F/\E_k \right)^2 - \frac{q^2 |\kappa^2 \mathfrak{F}^2|}{\E_k^2}.
\end{equation}
Thus, $\lambda$ should be larger or equal to zero. In the limit $q^2\mathfrak{F}^2=0$ we already have a diagonalized matrix with eigenvalues $\chi-\chi_\tau$ and $\chi_\tau/(n-2)$. If $q^2\mathfrak{F}^2=0$, there are only one direction with a distinguished principal curvature, otherwise there are two such directions.

{\bf Master equation.} In order to prove that a given surface is a massive particle surface, one should show that the surface is umbilic with respect to the directions $\tau^\alpha_{(i)}$ normal to $\kappa^\alpha$, i.e., show that Eq. (\ref{eq:chi_tt}) holds. However, in a number of certain calculations, it is convenient to rewrite remaining Eqs. (\ref{eq:chi_kk}) and (\ref{eq:chi_kt}) through some other known functions \cite{Kobialko:2021aqg,Kobialko:2022ozq}. First, Eq. (\ref{eq:chi_kk}) can be rewritten as
\begin{align}
    - \kappa^{-2} \kappa^\alpha n^\beta \nabla_\alpha \kappa_\beta = \SFS - \SFS_\tau.
\end{align}
After an explicit substitution of the relationship between $\SFS$ and $\SFS_\tau$ from Eq. (\ref{eq:definitions_3}), we find the master equation for the massive particle surfaces
\begin{align} \label{eq:master}
    - \kappa^{-2} \kappa^\alpha n^\beta \nabla_\alpha \kappa_\beta =
    \frac{1 + (m/\E_k)^2\kappa^2}{n-2}\SFS_\tau +q\F/\E_k.
\end{align}
Resolving Eq. (\ref{eq:master}) with respect to $\E$, one can find the total energy of the massive particle surface $\SE$ explicitly
\begin{equation}\label{eq:energy}
    \E_\pm =
    \pm m \sqrt{
              \frac{ \kappa^2 \SFS_\tau}{K}
            + \frac{\mathcal{F}^2 (n-2)^2 q^2}{4m^2K^2}
        }
    + \frac{\mathcal{F} (n-2) q}{2 K}
    - qk_\alpha A^\alpha,
\end{equation}
where
\begin{equation}\label{eq:K_def}
    K = - \SFS_\tau - (n-2) \kappa^{-2} \kappa^\alpha n^\beta \nabla_\alpha \kappa_\beta.
\end{equation}
The only physical branch with future-directed particles is $\E_+$. The right-hand side of Eq. (\ref{eq:energy}) must be constant on the surface, otherwise the surface under consideration is not a massive particle surface. Expression (\ref{eq:energy}) manifests the charge--time-reversal symmetry: $\E_\pm\to-\E_\mp,\,q\to-q$. The last condition for the massive particle surface to exist follows form Eq. (\ref{eq:chi_kt}) in the form 
\begin{align}\label{eq:offdiagonal}
    n^\alpha \tau_{(i)}^\beta \left(
        \kappa^{-2}  \LCM_\beta\kappa_\alpha + \frac{1}{2}(q/\E_k) F_{\alpha \beta}
    \right) = 0.
\end{align}

If the Killing vector is tangent to the surface, $k^\alpha = \kappa^\alpha$, the quantity $\nabla_\alpha \kappa_\beta$ is antisymmetric, and Eq. (\ref{eq:K_def}) can be simplified as follows
\begin{equation}
    K = - \SFS_\tau + \frac{n-2}{2} n^\alpha \nabla_\alpha \ln \kappa^2.
\end{equation}

{\bf Special cases.} In the case of neutral particles, the second fundamental form is block-diagonal and surface always has well-defined principal curvatures. However, the charged particle surfaces may also inherit this property, if the additional constraints on the Maxwell form is imposed, namely, $n^\alpha F_{\alpha \beta}\tau^\beta_{(i)}=0$. As we will see in examples from Sec. \ref{sec:examples}, this case is very common among physically interesting solutions. Since $\tau^\lambda_{(i)}$ is any vector from the corresponding subspace, the quantity $n^\alpha F_{\alpha \beta}$ should be parallel to $\kappa_\beta$:
\begin{align} \label{eq:f_parallel}
    n^\alpha F_{\alpha \beta}=f \kappa{}_{\beta}
    \quad \Rightarrow \quad
    \F_{\alpha\beta}=f \kappa_\alpha \kappa_\beta,
    \quad
    \F=f \kappa^2. 
\end{align} 
The condition from Eq. (\ref{eq:offdiagonal}) reduces to the simpler one:
\begin{align}
    n^\alpha \tau_{(i)}^\beta \LCM_\beta\kappa_\alpha = 0.
\end{align}
In this case, the electromagnetic Lorenz force manifests only through the definition of the kinetic energy $\E_k = \E + q k^\alpha A_\alpha$ and the relation between $\chi$ and $\chi_\tau$ in Eq. (\ref{eq:definitions_3}). If condition (\ref{eq:f_parallel}) holds, the traceless part $\sigma_{\alpha\beta}$ is identically zero if
\begin{equation}
    \frac{q\E_k}{m^2} = -\cfrac{\SFS}{(n-1) f}.
\end{equation}
Such surfaces are totally umbilic, which are known to be photon surfaces. The coincidence of a massive particle surface with a photon surface for some parameters has a simple physical explanation. Imagine a photon surface and throw massive neutral particles from this surface (with the speed smaller than the speed of light). As they are neutral, the electromagnetic field does not manifest in their dynamic and they will fall inside. However, for charged particles, one can try to find such a charge or energy that the electromagnetic field repels them as strongly as gravity attracts.

{\bf Gauge transformations.} Under a gauge transformation of the potential  $A_\alpha\rightarrow A_\alpha+\LCM_\alpha \varphi$, the potential energy is transformed in the following way
\begin{align}
\E_p\rightarrow\E_p+\delta \E_p, \quad \delta \E_p=-qk^\alpha \LCM_\alpha \varphi.
\end{align}
Imposing the requirement of the symmetry of the electromagnetic field 
\begin{align}
0=k^\lambda \LCM_\lambda \LCM_\alpha \varphi+  \LCM_\alpha k^\lambda  \LCM_\lambda \varphi=\LCM_\alpha k^\lambda \LCM_\lambda \varphi+k^\lambda \LCM_\lambda \LCM_\alpha \varphi=\LCM_\alpha (k^\lambda \LCM_\lambda \varphi), 
\end{align}
total and potential energies of the worldline can only shift by a constant under symmetry-preserving gauge transformations. Accordingly, the complete massive particle surfaces foliations are invariant up to a constant shift of the defining parameter.

{\bf Worldline stability.} As mentioned above, sets of fundamental photon surfaces form manifold foliation parametrized by values of the impact parameter. Similarly, in the massive case, a set of massive particle surfaces with different energies also forms some spacetime foliation. This foliation can be used, in particular, to obtain a very simple and physically intuitive condition for the stability of worldlines on a surface based on equivalent definition ($ii$)\footnote{The stability of worldlines can also be investigated by the method proposed in \cite{Koga:2020gqd,Koga:2019uqd} without using the foliation of massive particle surfaces with different energies, but the hypersurfaces flow of the same energy and the equivalent definition ($iii$). However, in this case the stability condition will explicitly depend on the choice of the worldline tangent vector.}. Consider a particle with a fixed mass $m$, charge $q$, total energy $\E_v = \E_{vk} + \E_p$, and some four-velocity $v^\alpha$, which is weakly perturbed from the worldline on the massive particle surface $\SE$. The equation of motion projected onto the normal direction is
\begin{align} 
    0 & = n_\beta v^\alpha \LCM_\alpha v^\beta  - q n_\beta {F^\beta}_\lambda v^\lambda =
    v^\alpha \LCM_\alpha v_n - v^\alpha v^\beta \nabla_\alpha n_\beta  - q n_\beta {F^\beta}_\lambda v^\lambda
     = \\\nonumber & =
    - v_n v_\tau^\beta n^\alpha \nabla_\alpha n_\beta
    + v^\alpha \LCM_\alpha v_n
    - v_\tau^\alpha v_\tau^\beta \SFS_{\alpha\beta}
    - q n_\beta {F^\beta}_\lambda v_\tau^\lambda
    + \mathcal{O}(v_n^2),
\end{align}
where $v^\alpha = v_\tau^\alpha + n^\alpha v_n$ is a decomposition of the vector onto the normal direction and tangent to the hypersurface. Since the worldline represents a small deviation from the worldlines from the surface $\SE$, we consider $v_n$ is small. The quantity $v^\alpha \LCM_\alpha v_n$ represents the normal acceleration along the worldline
\begin{align}
    &
      v^\alpha \LCM_\alpha v_n =
      v_\tau^\alpha v_\tau^\beta \SFS_{\alpha\beta}
    + q n_\beta {F^\beta}_\lambda v_\tau^\lambda
    + v_n v_\tau^\beta n^\alpha \nabla_\alpha n_\beta
    + \mathcal{O}(v_n^2) \approx a_{n},
\end{align}
where we have introduced the quantity $a_n$
\begin{equation}
    a_n = \SFS_{\alpha\beta}v_\tau^\alpha v_\tau^\beta+q n^\beta F_{\beta\lambda} v_\tau^\lambda + v_n v_\tau^\beta n^\alpha \nabla_\alpha n_\beta.
\end{equation}
In the limit of the worldlines lying on the surface $\SE$ (i.e., $\E_v=\E$, $v_n=0$), the normal acceleration is zero $a_n=0$. Substituting the second fundamental form of a massive particle surface with energy $\E$ and an arbitrary vector $v^\alpha$, we get
\begin{equation}
    a_n = 
      \frac{\SFS_\tau}{n-2}(-m^2-(n_\alpha v^\alpha)^2+(m\E_{vk}^{\parallel}/\E_k)^2)
      - (q\E_{vk}^{\parallel}/\E_k) n^\rho F_{\rho \lambda} v^\lambda
      + q n^\beta F_{\beta\lambda} v^\lambda
      + v_n v_\tau^\beta n^\alpha \nabla_\alpha n_\beta,
\end{equation}
where $\E_{vk}^{\parallel} \equiv \E_{vk} + k_\perp (n_\alpha v^\alpha)$ is the kinetic energy of the tangential motion. The change of the normal acceleration along the worldline is described by $v^\alpha \nabla_\alpha a_n$:
\begin{align}
    \nabla_v a_n &= 
      \frac{\nabla_v\SFS_\tau}{n-2}(-m^2-(n_\alpha v^\alpha)^2+m^2(\E_{vk}^{\parallel}/\E_k)^2)
    \\\nonumber &
    + \frac{\SFS_\tau}{n-2}(
        - 2 (n_\beta v^\beta) \nabla_v(n_\alpha v^\alpha)
        +m^2\nabla_v(\E_{vk}^{\parallel}/\E_k)^2)
    \\\nonumber &
    +(1 - \E_{vk}^{\parallel}/\E_k) q \nabla_v\left(n^\rho F_{\rho \lambda} v^\lambda\right)
    - q\nabla_v(\E_{vk}^{\parallel}/\E_k) n^\rho F_{\rho \lambda} v^\lambda
    + \nabla_v (v_n v_\tau^\beta n^\alpha \nabla_\alpha n_\beta),
\end{align}
where $\nabla_v \equiv v^\alpha \nabla_\alpha$. Here, we consider $\E$ (and $\E_k$) is a function of the coordinates, since each hypersurface in the foliation has own value, while $\E_v$ is a constant energy of the particle with four-velocity $v^\alpha$, so
\begin{equation}\label{eq:dE1}
    \nabla_v(\E_{vk}^{\parallel}/\E_k) = 
    - \E_k^{-1} \nabla_v\E = - v_n \E_k^{-1} \nabla_n\E,
\end{equation}
where we used the constancy of $\E$ on the surface and omit higher corrections in $v_n$ and $1 - \E_{vk}^{\parallel}/\E_k$. Keeping terms which are only linear in $v_n$ and $1 - \E_{vk}^{\parallel}/\E_k$, we get
\begin{align} \label{eq:dF1}
    \nabla_v a_n &\to (\nabla_v a_n)|_{\SE} = 
    - v_n \E_k^{-1} \left(
    \frac{2 m^2 \SFS_\tau}{n-2}
    - q n^\rho F_{\rho \lambda} v^\lambda
    \right) \nabla_n\E
    \\\nonumber &
    + \left(1 - \E_{vk}^{\parallel}/\E_k\right)
    v_\tau^\alpha \nabla_\alpha \left(
        - \frac{m^2}{n-2} \SFS_\tau
        + q n^\rho F_{\rho \lambda} v^\lambda
    \right)
    + v_\tau^\alpha \nabla_\alpha (v_n v_\tau^\beta n^\alpha \nabla_\alpha n_\beta).
\end{align} 
Eq. (\ref{eq:dF1}) allows to analyze the stability of the worldline. In practice, one can often find that $n^\alpha \nabla_\alpha n_\beta = 0$ \cite{Koga:2020gqd}, and quantities $\SFS_\tau$ and $n^\rho F_{\rho \lambda} v^\lambda$ are constant on the hypersurface. The term $n^\rho F_{\rho \lambda} \tau_{(i)}^\lambda$ describes a non-central magnetic field, which is not plausible in the spacetimes with massive particle surfaces, and we consider this term equal to zero. Using these facts and decomposition (\ref{eq:v_decomposition}), in practice we will use the following expression
\begin{align}\label{eq:dF2}
    \nabla_v a_n &= 
    v_n W_n \E_k^{-1} \nabla_n\E, \qquad
    W_n \equiv
        - \frac{2 m^2 \SFS_\tau}{n-2}
        - q \F \E_k / \kappa^2.
\end{align}
For convex hypersurfaces ($\chi_\tau > 0$), in the chargeless case, the constant $W_n$ is always negative.

This result has a number of advantages over the stability formulas from \cite{Koga:2020gqd,Koga:2019uqd}. First, Eq. (\ref{eq:dF2}) does not contain an arbitrary vector tangent the worldlines, i.e., it characterizes the stability of all worldlines and consequently the massive particle surface by itself. In particular, we can define analogs of antiphoton and photon surfaces \cite{Gibbons:2016isj}. Secondly, the use of this equation does not require the calculation of any new quantities such as the Riemann tensor and so on, except for the derivative of the previously determined surface energy.  Furthermore, this result has a demonstrative physical explanation. Let the foliation is parameterized by some parameter $r$. According to Eq. (\ref{eq:energy}), the total energy $\E$ is a function of the foliation parameter: $\E=\E(r)$. Consider that we analyze a metric, where one can separate the radial equation as $\dot{r}^2 = R(r, L, r)$, where $L$ is a set of other integrals of motion (e.g., total angular momentum), which are also fixed at each $\SE$. Then, $\E(r)$ corresponds to the maximum or minimum of $R$ for some $L$ placed at $r$. Maximums corresponds to unstable orbits and minimums corresponds to stable orbits. If the maximum merges with another minimum at some $r$, such worldlines represent marginally stable orbits (such as innermost stable circular orbit -- ISCO). In a physically meaningful case, the maximum and the minimum meet each other at their highest or lowest values (see an example for Schwarzschild solution in Fig. \ref{fig:schwarzschild}). Otherwise, two curves $R(\E,L_1,r)$ and $R(\E,L_2,r)$ representing effective radial potentials for different $L$'s would intersect each other, which is not a typical physical case. So, $d\E/dr = 0$ distinguishes the marginally stable orbits, separating stable and unstable orbits.

\section{Static examples}
\label{sec:examples}
Let the static four-dimensional spacetime has the following form of the metric tensor
\begin{equation}\label{eq:metric_ansatz}
    ds^2 = -\alpha dt^2 + \gamma d\phi^2 + \lambda dr^2 + \beta d\theta^2,
\end{equation}
where $\alpha, \beta, \lambda, \gamma$ are functions of $r,\theta$ and we choose a surface $r = \text{const}$. Then the second fundamental form is
\begin{equation}
    \chi_{\alpha\beta} = \frac{1}{2\sqrt{\lambda}}\left(
        -\partial_r \alpha dt^2 + \partial_r \gamma d\phi^2 + \partial_r \beta d\theta^2
    \right),\qquad
    \chi = \chi_\alpha^\alpha = \frac{\partial_r \ln (\alpha\beta\gamma)}{2\sqrt{\lambda}},
\end{equation}
and we choose the Killing vector along the asymptotic time coordinate $k^\alpha\partial_\alpha=\partial_t$ (i.e., $\kappa^2 = -\alpha$ and $\E$ represents the total energy of the particle). The Maxwell-related form $\F_{\alpha\beta}$ reads
\begin{equation}
    \mathcal{F}_{\alpha\beta}dx^\alpha dx^\beta = -\frac{\alpha}{\sqrt{\lambda}}\left(
        \partial_r A_t dt + \partial_r A_\phi d\phi 
    \right)dt,\qquad
    \mathcal{F} = \mathcal{F}_{\alpha}^{\alpha} = \frac{\partial_r A_t}{\sqrt{\lambda}}.
\end{equation}

\subsection{Schwarzschild}
The simplest example is the Schwarzschild black hole with mass $M$ defined by the following vacuum metric
\begin{equation}
    ds^2 = -f dt^2 + f^{-1} dr^2 + r^2(d\theta^2 + \sin^2\theta d\phi^2),\qquad f = \frac{r-2M}{r},
\end{equation}
with surfaces $r=\text{const}$ and the Killing vector $k^\alpha\partial_\alpha=\partial_t$. Indeed, spheres are umbilic with the mean curvature
\begin{equation}
    \chi_\tau = \frac{2f^{1/2}}{r}.
\end{equation}
The second fundamental form is diagonal. Substituting other necessary quantities
\begin{align}
    &
    \chi = f^{-1/2}\frac{2 r-3 M}{r^2},\quad
    \kappa^2 = -f,\quad
    K = 2f^{-1/2}\frac{-r+3M}{r^2},\quad
    \F = 0,\quad
    k^\alpha A_\alpha = 0
\end{align}
in Eq. (\ref{eq:energy}) we find the following expression for the energy of massive particle surfaces
\begin{equation}\label{eq:energy_sch}
    \E^2 / m^2 = \frac{(r - 2 M)^2}{r (r - 3 M)}.
\end{equation}
Innermost stable circular orbits (ISCO) defined from the condition $d\E/dr = 0$ has the radius $r_{\text{ISCO}}=6M$ with $\E^2 / m^2 = 8/9$. Expression (\ref{eq:energy_sch}) diverges at $r_\text{ph}=3M$ which corresponds to the photon surface (see Ref. \cite{Claudel:2000yi}).

Let us compare this result with the dynamical approach. The conserving total energy and total angular momentum read
\begin{equation}
    \E = -\dot{\gamma}^\alpha k_\alpha = f \dot{\gamma}^t,\qquad
    L^2 = r^4\left[\left(\dot{\gamma}^{\theta }\right)^2 + \left(\dot{\gamma}^{\phi}\sin\theta\right)^2\right].
\end{equation}
Substituting these into the equation $\dot{\gamma}^2=-m^2$, we get the radial equation
\begin{equation}\label{eq:potential_sch}
    \dot{r}^2 = R(\E, L, r) \equiv \E^2 - \frac{r-2 M}{r^3} \left(L^2 + m^2 r^2\right).
\end{equation}
The conditions $\dot{r}=0$ and $\ddot{r}=0$ (i.e., $R=0$ and $\partial_r R = 0$) leads to the following integrals of motion as a function of the sphere radius
\begin{equation}
    \E^2/m^2 = \frac{(r-2 M)^2}{r (r-3M)},\qquad
    L^2/m^2 = \frac{M r^2}{r-3 M}.
\end{equation}
The condition that distinguishes ISCO, $\partial_r^2 R = 0$, approves the previous result $r_{\text{ISCO}}=6M$. The angular momentum of ISCO orbits is $L=2\sqrt{3}mM \approx 3.46 mM$.

The squared total energy function (\ref{eq:energy_sch}) and $R(0,L,r)$ from Eq. (\ref{eq:potential_sch}) for different $L$ are given in Fig. \ref{fig:schwarzschild}. Function $-\E^2/m^2$ goes through maxima and minima of $R(0,L,r)$ for certain values of $L$. It is defined for $r>3M$, increasing in the region of unstable orbits $3M < r < 6M$ and decreasing in the region of the stable orbits $r > 6M$. At $r=6M$ it has a maximum $-8/9$ corresponding to marginally stable orbits (namely, ISCO).

\begin{figure*}[tb!]
    \centering
        \includegraphics[width=0.6\textwidth]{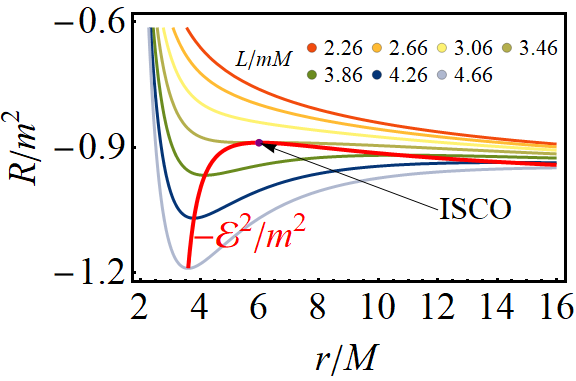}
          \caption{The squared total energy (\ref{eq:energy_sch}) and $R(0,L,r)$ from Eq. (\ref{eq:potential_sch}) as a function of $r$ for different $L$ in Schwarzschild spacetime.\label{fig:schwarzschild}}
\end{figure*}

\subsection{Newman-Unti-Tambourino}
The solutions of vacuum general relativity with mass $M$ and gravitomagnetic Newman-Unti-Tambourino charge $N$ (NUT) possess Killing vectors forming an algebra of a spherical symmetry and time translations, though it is not evident from the form of the metric
\begin{equation}
    ds^2 = -f (dt+2N\cos\theta d\phi)^2 + f^{-1} dr^2 + (r^2+N^2)(d\theta^2 + \sin^2\theta d\phi^2),\qquad f = \frac{r(r-2M)-N^2}{r^2+N^2}.
\end{equation}
The algebra of the spherical symmetry can be exponentiated up to the group, if one endows surfaces $r=\text{const}$ with the topology of three-spheres $\mathds{S}^3$ \cite{Misner:1963fr}, forcing the time to be compact. This metric is not diagonal and does not fit ansatz (\ref{eq:metric_ansatz}). However, one can check that the surface $r=\text{const}$ is umbilic in the sector orthogonal to the Killing vector $\partial_t$, i.e., $\tau_{(1)}^\alpha\partial_\alpha = \partial_\theta$, $\tau_{(2)}^\alpha\partial_\alpha = \partial_\phi-2N\cos\theta\partial_t$. The mixed components $\kappa^\alpha\chi_{\alpha\beta}\tau^\beta_{(i)}$ are not presented. Thus, one can apply the master equation from Eq. (\ref{eq:energy}). Performing the same steps as for the Schwarzschild metric, we get
\begin{subequations}
\begin{equation}
    \chi_\tau = \frac{2 r}{r^2 + N^2}f^{1/2},
\end{equation}
\begin{equation}
    \chi =f^{-1/2} \frac{2 r^3-M \left(N^2+3 r^2\right)}{\left(N^2+r^2\right)^2},
\end{equation}
\begin{equation}
    K = 2f^{-1/2}\frac{-r^3 + 3 M r^2 + 3 N^2 r - M N^2 }{\left(N^2+r^2\right)^2},
\end{equation}
\begin{equation}\label{eq:energy_nut}
    \E^2 / m^2 = \frac{
        r \left(r (2 M-r)+N^2\right)^2
    }{
        \left(N^2+r^2\right) \left(r^3 -3 M r^2 - 3 N^2 r + M N^2\right)
    }.
\end{equation}
\end{subequations}
The equation on the marginally stable orbits is a polynomial equation of degree six and cannot be resolved explicitly. Function (\ref{eq:energy_nut}) is depicted in Fig. \ref{fig:nut} for different values of $N$. Also, there is a curve $(r_\text{ISCO},\E_\text{ISCO})$ parameterized by $N$ found numerically.

\begin{figure*}[tb!]
    \centering
        \includegraphics[width=0.6\textwidth]{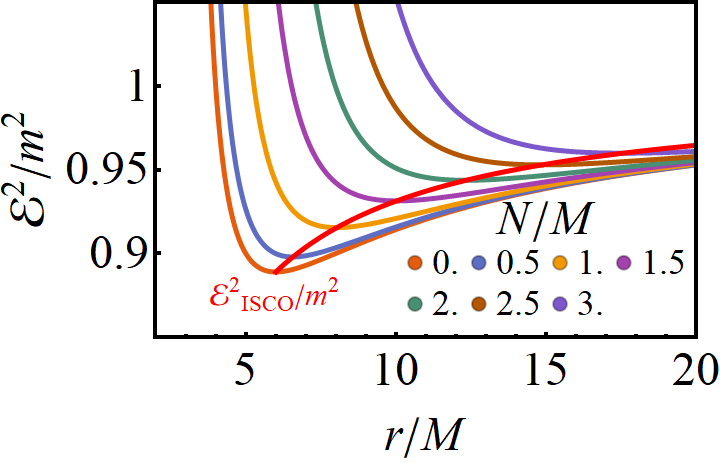}
          \caption{The squared total energy (\ref{eq:energy_nut}) as a function of $r$ and squared energy corresponding to ISCO for different $N$ in NUT spacetime.\label{fig:nut}}
\end{figure*}

\subsection{Fisher-Janis-Newman-Winicour}
The Fisher-Janis-Newman-Winicour solution (FJNW) is a Schwarzschild generalization with a scalar charge $\Sigma$ \cite{Ota:2021mub}. The FJNW metric reads
\begin{equation}
    ds^2 = -f^\sigma dt^2 + f^{-\sigma} dr^2 + f^{1-\sigma}r^2(d\theta^2 + \sin^2\theta d\phi^2),\qquad
    f = 1 - \frac{2M}{\sigma r},\qquad
    \sigma = \frac{M}{\sqrt{M^2 + \Sigma^2}}.
\end{equation}
One can restore Schwarzschild solution with $\sigma=1$. For $0<\sigma<1$ the solution is singular at $r=2M/\sigma$. Similarly, we consider surfaces $r=\text{const}$ with a Killing vector $\partial_t$. The total energy function (\ref{eq:energy}) of the massive particle surfaces is defined by
\begin{equation}\label{eq:energy_fisher}
    \E^2/m^2 = \left(1-\frac{2 M}{r \sigma }\right)^{\sigma } \frac{r \sigma - M (\sigma + 1)}{r \sigma - M (2 \sigma + 1)}.
\end{equation}
Differentiating the expression for $\E^2/m^2$, one will find two extrema at
\begin{equation}
    r_\pm = M \frac{3 \sigma +1  \pm \sqrt{ 5 \sigma^2 - 1 }}{\sigma},
\end{equation}
where the sign $+$ ($-$) stands for a minimum (maximum). At $\sigma = 1/\sqrt{5}$ the minimum and the maximum merge and disappear (Fig. \ref{fig:fisher}). The photon surface, defined by the divergence of the expression in Eq. (\ref{eq:energy_fisher}), is placed at $r_{\text{ph}}=M(2\sigma+1)/\sigma$ for $1/2\leq \sigma \leq 1$ \cite{Claudel:2000yi}. Thus, we have four cases:
\begin{enumerate}
    \item $1/2 \leq \sigma \leq 1$. There is a photon surface at $r_{\text{ph}}=M(2\sigma+1)/\sigma$ and ISCO at $r_\text{ISCO}=r_+$. Surfaces $r_\text{ph} \leq r < r_\text{ISCO}$ are unstable and $r_\text{ISCO} < r$ are stable.
    \item $1/\sqrt{5} < \sigma < 1 / 2$. There are no photon surfaces and two marginally stable orbits $r_\pm$ exist. Stable orbits are placed at $2M/\sigma < r < r_-$ and $r_+ < r$ and unstable orbits are placed at $r_- < r < r_+$.
    \item $0\leq\sigma=1/\sqrt{5}$. There are no photon surfaces, and a degenerate marginally stable orbit at $r=M(3+\sqrt{5})$, and for other $2M\sqrt{5} < r$ orbits are stable.
    \item $0\leq\sigma<1/\sqrt{5}$. There are no photon surfaces and marginally stable orbits. All orbits $2M/\sigma < r$ are stable.
\end{enumerate}

\begin{figure*}[tb!]
    \centering
        \includegraphics[width=0.6\textwidth]{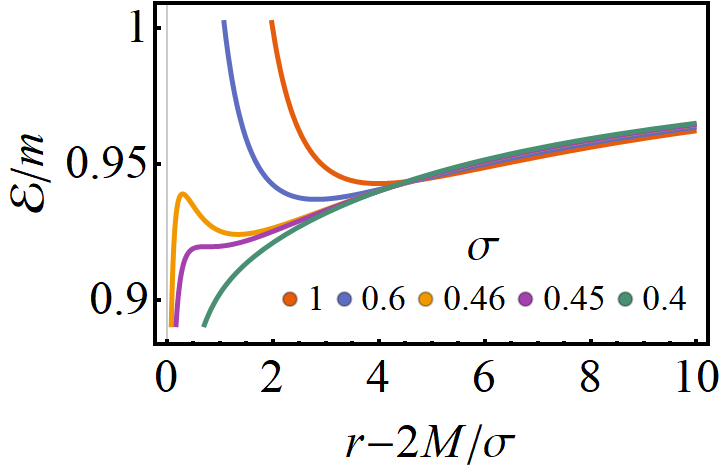}
        \caption{The total energy (\ref{eq:energy_fisher}) as a function of $r$ for different $\sigma$ in FJNW spacetime.\label{fig:fisher}}
\end{figure*}

\subsection{EMD dyons with stable photon spheres}
Einstein-Maxwell-dilaton (EMD) dyons with mass $M$, NUT $N$, scalar charge $D$ and electric and magnetic charges $Q$ and $P$ where revisited in Ref. \cite{Bogush:2020obx}, where one can find the full definition of the solution. The scalar charge $D$ is constrained by a cubic equation for regular solutions. Let us consider neutral massive particles placed at the surface $r=\text{const}$ with a total energy defined with respect to the Killing vector $\partial_t$. In Ref. \cite{Bogush:2020obx} stable photon surfaces are found for some special classes of the naked singularities of the theory. Stable photon surfaces indicates that the solution is unstable. We will analyze these solution in order to trace the behavior of the massive particle surfaces in spacetimes with stable photon surfaces. For example, we will consider solutions with the following charges
\begin{equation}
    N/M=0,\qquad
    Q/M = -1.5 \cos\epsilon \approx -1.5,\qquad
    P/M = 1.5 \sin\epsilon \approx 1.5 \epsilon
\end{equation}
with some small $\epsilon = 0.0998,\,0.11,\,0.125$ (Fig. \ref{fig:emd}). For the first two values $\epsilon = 0.0998,\,0.11$ we have a similar structure (orange and blue curves). Namely, we have 2 unstable photon surfaces (let us denote them as $r_{\text{ph}}^{\pm \text{u}}$ with $r_{\text{ph}}^{- \text{u}} < r_{\text{ph}}^{+ \text{u}}$) and one stable in-between (let it be $r_{\text{ph}}^\text{s}$). Also, there are two marginally stable obits $r_\text{min}^{1,2}$ such that $r_{\text{ph}}^{- \text{u}} < r_{\text{min}}^{1} < r_{\text{ph}}^{\text{s}}$ and $r_{\text{ph}}^{+ \text{u}} < r_{\text{min}}^{2}$. In the region $r_{\text{ph}}^{\text{s}} < r < r_{\text{ph}}^{+ \text{u}}$ there are no massive particle surfaces at all. In the third case $\epsilon = \,0.125$ (yellow curve), two photon orbits $r_{\text{ph}}^{\text{s}}$, $r_{\text{ph}}^{+ \text{u}}$ disappear, and one more marginally stable orbit $r_\text{max}$ appears, corresponding to the maximum of the yellow curve. It is notable that near the stable photon surface $r_{\text{ph}}^{\text{s}}$, the energy of the surface runs from its minimal value to $+\infty$ in a very narrow interval of $r$, populating this region with stable surfaces with all possible energies. For the smaller values of $\epsilon$, the solution provides itself with an event horizon, covering the stable photon surface $r_{\text{ph}}^{\text{s}}$.
\begin{figure*}[tb!]
    \centering
        \includegraphics[width=0.6\textwidth]{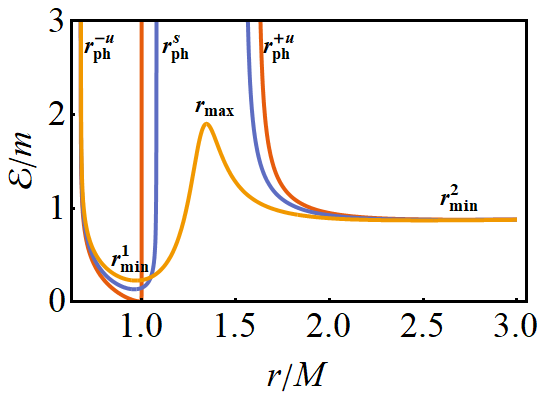}
          \caption{The total energy (\ref{eq:energy}) as a function of $r$ in EMD dyon without NUT $N=0$ for different values of $\epsilon$: (orange) 0.0998, (blue) 0.11, and (yellow) 0.125. \label{fig:emd}}
\end{figure*}

\subsection{Reissner-Nordstr\"{o}m dyon}
The dyonic Reissner-Nordstr\"{o}m solution with mass $m$, electric charge $Q$ and magnetic charge $P$ reads
\begin{equation}
    ds^2 = -\frac{\Delta}{r^2} dt^2 + \frac{r^2}{\Delta} dr^2 + r^2(d\theta^2 + \sin^2\theta d\phi^2),\qquad
    \Delta = r(r-2M) + Q^2 + P^2,
\end{equation}
\begin{equation}
    A_\mu dx^\mu = \frac{Q}{r} dt + P \cos\theta d\phi.
\end{equation}
Applying Eq. (\ref{eq:energy}), the solution for $\E$ reads
\begin{equation}
    \E_\pm/m = 
    -\frac{qQ}{2 m r}
    + \frac{
            (qQ/m) \left( M r - P^2 - Q^2 \right)
            \pm 2 \Delta \sqrt{ \Delta - M r + P^2 + Q^2 + (qQ/2m)^2 }
        }{
            2 r \left(r^2-3 M r+2 P^2+2 Q^2\right)
        }.
\end{equation}
The photon sphere is placed at the root of the denominator \cite{Claudel:2000yi} 
\begin{equation}
    r_{\text{ph}}=\frac{1}{2} \left(3 M + \sqrt{9 M^2-8 (Q^2+P^2)}\right).
\end{equation}
The marginally stable orbits determined from the condition $d\E_+/dr = 0$ are depicted in Fig. \ref{fig:rn}. The structure of the corresponding curves differs for subextremal and superextremal cases. In the degenerate case, the curve is tangent to the vertical line, and $d^2\E/dr^2 = 0$. The region in the upper left corner is the region of the stable coordinates. The left upper corner is characterized by the absence of the stability. Thus, the charges corresponding to these values cannot be stable at any distance from the center of the solution. The straight line at $qQ/mM=1$ with $M^2 = P^2 + Q^2$ corresponds to the Bogomol'nyi–Prasad–Sommerfield (BPS) solutions \cite{Scherk:1979aj,Ortin:2015hya}, which are known to satisfy the no-force condition. In this case, the test particle and the central object do not interact with each other, and the energy of the surfaces are the same for any radius $r$. The curve $\sqrt{(P^2 + Q^2)/M^2} = 0$ is not a straight horizontal line, because it is understood as a limit $Q \to 0$ with a finite value of $qQ$.

\begin{figure*}[tb!]
    \centering
        \includegraphics[width=0.6\textwidth]{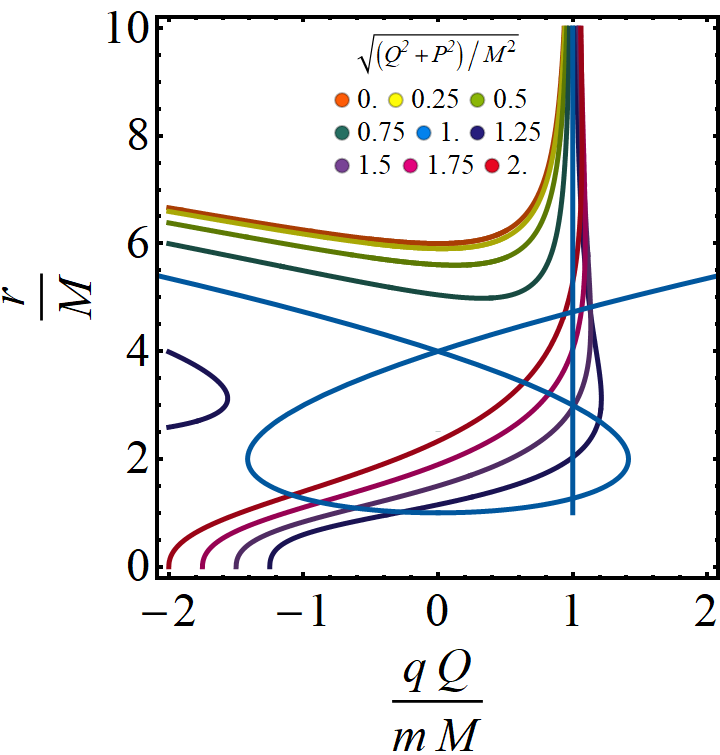}
          \caption{The position of the marginally stable orbits as a function of $qQ/mM$ for different values of $\sqrt{(Q^2 + P^2)/M^2}$ in Reissner-Nordstr{\"o}m spacetime. \label{fig:rn}} 
\end{figure*}

\section{Conclusions}
We proposed a generalization of the photon  surfaces introduced by Claudel, Virbhadra, and Ellis to the case of massive charged particles. An important new feature of the massive case is that the conformal non-invariance of the corresponding worldlines makes it necessary to restrict particles with fixed energy, assuming the existence of the corresponding Killing symmetry. With this in mind, we define new surfaces of massive particles as timelike hypersurfaces such that any worldline of a particle with mass $m$, electric charge $q$, and total energy $\E$, initially touching the surface, will remain tangent to it forever.

We have established the key theorem \ref{th:theorem}, which is a complete analog of theorem 2.2 obtained in Ref. \cite{Claudel:2000yi} for photon surfaces. The most important point of the theorem is the statement ($iii$) which describes the pure geometry of a massive particle surface without references to the worldline dynamic equations and it represents a modification of the totally umbilic condition for photon surfaces. Such equivalent definition of massive particle surfaces is an effective way to analyze their geometry for non-integrable dynamical systems and to study general theoretical problems such as the possibility of constructing Penrose inequalities for massive particle surfaces spatial sections and uniqueness theorems for asymptotically flat spacetimes with massive particle surfaces. 

Furthermore, we have established that statement ($iii$) is nothing else than the condition of hypersurface partial umbilicity, i.e., exactly the same property that defines the fundamental photon surfaces in stationary spacetimes. The only difference between these cases is that instead of an impact parameter, the geometry of massive particle surfaces is determined by energy, and the presence of an electromagnetic field imposes additional conditions on the mixed components of the second fundamental form. In addition, massive particle surfaces usually have no boundaries. The similarity between fundamental photon and massive particle surfaces also lies in the fact that, just as the former form photon regions, the latter also form foliations of a spacetime locally parameterized by energy values.

We have derived a master equation governing the energy of the surface. In the case of neutral particles, the energy of the surface depends on the Killing vector projection and the mean curvature of the surface in directions orthogonal to the Killing vector. For charged particles, the energy of the surface also includes the force of the electric field acting on the particle and the potential energy of the particle. We have found that photon surfaces can be also massive particle surfaces for repelling charges.  

The condition of the stability of worldlines lying on the massive particle surface was derived. In practice, it appeals to differentiation of the surface energy along the flow of the massive particle surfaces and does not depend on the velocities of individual particles lying in them in practice. In particular, an extremum of the energy corresponds to marginally stable orbits.

We have considered a number of examples of well-known electrovacuum and Einstein-Maxwell-dilaton solutions, demonstrating the application of the the developed instrument. It was shown to be helpful for finding marginally stable orbits, regions of stable or unstable spherical orbits, stable and unstable photon surfaces, and solutions satisfying the no-force condition.  

We hope that massive particle surfaces, their geometric definition, and possible generalizations to the case of spacetimes with Kerr-type rotation will be useful in a variety of theoretical applications. In particular, we expect that it will be possible to obtain Penrose inequalities, uniqueness theorems, as well as to establish connections with explicit and hidden symmetries of spacetimes, integrability, and mechanisms for the formation of shadows created by scattering of massive charged particles.
\begin{acknowledgments}
The work was supported by the Russian Foundation for Basic Research on the project 20-52-18012 Bulg-a, and the Scientific and Educational School of Moscow State University ``Fundamental and Applied Space Research''. I.B. is also grateful to the Foundation for the Advancement of Theoretical Physics and Mathematics ``BASIS'' for support.  
\end{acknowledgments}

\appendix
\section{Case \texorpdfstring{$\kappa^2 = 0$}{}}
\label{sec:kh}

In case $\kappa^2=0$ (but $\kappa^\alpha \neq 0$), we can introduce the following decomposition 
\begin{align}
&v^\alpha=\tilde{\E}_k \kappa^\alpha + \E_k \Tilde{\kappa}^\alpha  + u^\alpha, \quad  \kappa_\alpha  u^\alpha=\Tilde{\kappa}_\alpha  u^\alpha=0, \quad \kappa_\alpha \Tilde{\kappa}^\alpha=-1, \quad  \Tilde{\kappa}_\alpha \Tilde{\kappa}^\alpha=0, \\
& u^2 - 2 \tilde{\E}_k \E_k=-m^2,\qquad -\tilde{\kappa}_\alpha v^\alpha \equiv \tilde{\E}_k = (m^2 + u^2)/(2\E_k).
\end{align}
The orthogonal complement can again contain only spacelike vectors $2 \tilde{\E}_k \E_k>m^2$ and in particular again $\E_k, \tilde{\E}_k\neq0$.

The second fundamental form can be decomposed into
\begin{align}
    &
    \SFS_{\alpha\beta} =
      \alpha_{++} \kappa_\alpha  \kappa_\beta 
    + \alpha_{+-} \kappa_{(\alpha}  \tilde{\kappa}_{\beta)} 
    + \alpha_{--} \tilde{\kappa}_\alpha  \tilde{\kappa}_\beta 
    + \kappa_{(\alpha} \beta_{\beta)}
    + \tilde{\kappa}_{(\alpha} \tilde{\beta}_{\beta)}
    +\lambda_{\alpha\beta}
    +(q/\E_k)\F_{\alpha\beta}, 
    \\\nonumber&
    \kappa^\alpha\lambda_{\alpha}{}_{\beta} =
    \tilde{\kappa}^\alpha\lambda_{\alpha}{}_{\beta} = 0, \quad
    \kappa^\alpha\beta_\alpha = \tilde{\kappa}^\alpha\beta_\alpha = 0,
\end{align}
where the last term in $\chi_{\alpha\beta}$ was introduced to compensate the right hand side in Eq. (\ref{eq:chi_condition}), giving the following condition
\begin{align}
      \alpha_{++} \E_k^2
    + 2\alpha_{+-} \E_k \tilde{\E}_k
    + \alpha_{--} \tilde{\E}_k^2
    - 2 \E_k \beta_{\alpha} u^\alpha
    - 2 \tilde{\E}_k \tilde{\beta}_{\alpha} u^\alpha
    +\lambda_{\alpha\beta} u^\alpha u^\beta
    = 0.
\end{align}
It should hold for any $\tilde{\E}_k$ and $u^\alpha$ satisfying the norm. From the $u^\alpha$-parity analysis, it splits into two parts
\begin{align}
    &
      \E_k \beta_{\alpha} u^\alpha
    + \tilde{\E}_k \tilde{\beta}_{\alpha} u^\alpha
    = 0,
    \\\nonumber &
      \alpha_{++} \E_k^2
    + 2\alpha_{+-} \E_k \tilde{\E}_k
    + \alpha_{--} \tilde{\E}_k^2
    + \lambda_{\alpha\beta} u^\alpha u^\beta
    = 0
\end{align}
Let us substitute $\tilde{\E}_k$ explicitly and perform a scaling transformation $u^\alpha \to k u^\alpha$:
\begin{subequations}
\begin{align}
    & \label{eq:null_condition_scaled_1}
      2\E_k^2 \beta_{\alpha} u^\alpha
    +(m^2+k^2 u^2) \tilde{\beta}_{\alpha} u^\alpha
    = 0,
    \\ & \label{eq:null_condition_scaled_2}
      \alpha_{++} \E_k^2
    + (m^2+k^2u^2)\alpha_{+-}
    + \alpha_{--} \frac{(m^2+k^2u^2)^2}{4\E_k^2}
    + k^2  \lambda_{\alpha\beta} u^\alpha u^\beta
    = 0.
\end{align}
\end{subequations}
Definitely, we can just take certain values of $k$ and then add and subtract equations with different $k$. Instead, we will differentiate the equation with respect to $k^2$. Differentiating Eq. (\ref{eq:null_condition_scaled_1}), we come to the conclusion that two terms are equal to zero separately $\beta_{\alpha} u^\alpha=\tilde{\beta}_{\alpha} u^\alpha=0$, and since $u^\alpha$ is arbitrary, each co-vector is zero: $\beta_\alpha=\tilde{\beta}_\alpha = 0$. Twice differentiating Eq. (\ref{eq:null_condition_scaled_2}) leads to $\alpha_{--}=0$. And differentiation of Eq. (\ref{eq:null_condition_scaled_2}) once gives
\begin{align}
      u^2 \alpha_{+-}
    + \lambda_{\alpha\beta} u^\alpha u^\beta
    = 0.
\end{align}
Similarly to the analysis for $\kappa^2\neq0$, we find the tensor $\lambda_{ab} = - \alpha_{+-}(h_{\alpha\beta} + \kappa_{(\alpha}\tilde{\kappa}_{\beta)})$. Substituting this back into Eq. (\ref{eq:null_condition_scaled_1}) we get $\alpha_{++} = -  \alpha_{+-} m^2 / \E_k^2$. Collecting all terms together gives the final expression for the second fundamental form
\begin{equation}
    \chi_{\alpha\beta} = 
    -\alpha_{+-} 
    \left(
          h_{\alpha\beta}
        + \frac{m^2}{\E_k^2} \kappa_\alpha  \kappa_\beta
    \right)
    +(q/\E_k)\F_{\alpha\beta}.
\end{equation}
After the redefinition of $\alpha_{+-}$, one will arrive to the same expression as for $\kappa^2\neq0$.

\end{document}